\documentclass[letterpaper, 10 pt, conference]{ieeeconf}  % Comment this line out if you need a4paper

\IEEEoverridecommandlockouts                              % This command is only needed if 
\usepackage{amsmath,amssymb}
\usepackage{graphicx,epstopdf}
\usepackage{color}
\usepackage{amsfonts}
\usepackage{stfloats}
\usepackage{bm}
\usepackage{algorithmic}

\newcommand{\R}{\mathbb{R}}

\newcommand{\C}{\mathcal{C}}

\newcommand{\la}{\lambda}
\newcommand{\rh}{{\rm H}_h}
\newcommand{\ga}{\gamma}

\allowdisplaybreaks[4]

\newif\ifdraft
\drafttrue

\newcommand{\fot}{\frac{1}{2}}

\newcommand{\ta}{\theta}
\newcommand{\pa}{\partial}

\newcommand{\ep}{\epsilon}

\newtheorem{definition}{\bfseries Definition}%[section]
\newtheorem{proposition}{\bfseries Proposition}%[section]
\newtheorem{assumption}{\it Assumption}%[section]
\newtheorem{theorem}{\bfseries Theorem}
\newtheorem{lemma}{\bfseries Lemma}%[section]
\newtheorem{property}{\bfseries Property}
\newtheorem{remark}{\bfseries Remark}

\newtheorem{problem}{\bfseries Problem}

\newtheorem{task}{\bfseries Task}

\makeatletter
\renewcommand\normalsize{%
\@setfontsize\normalsize\@xpt\@xiipt
\abovedisplayskip 1.9\p@ \@plus2\p@ \@minus3\p@
\abovedisplayshortskip \z@ \@plus3\p@
\belowdisplayshortskip 1.9\p@ \@plus3\p@ \@minus3\p@
\belowdisplayskip \abovedisplayskip
\let\@listi\@listI}
\makeatother

\title{\LARGE \bf
Safe Control of Euler-Lagrange Systems with Limited Model Information
}

\author{}

\author{Yujie Wang and Xiangru Xu\thanks{Y. Wang and X. Xu are with the Department of Mechanical Engineering, University of Wisconsin-Madison,
        Madison, WI 53706, USA. Email: 
\{yujie.wang, xiangru.xu\}@wisc.edu.}}

\begin{document}

\maketitle

%%%%%%%%%%%%%%%%%%%%%%%%%%%%%%%%%%%%%%%%%%%%%%%%%%%%%%%%%%%%%%%%%%%%%%%%%%%%%%%%

%---------------------------------------
%            Abstract                  |
%---------------------------------------

\begin{abstract}
This paper presents a new safe control framework for Euler-Lagrange (EL) systems with limited model information, external disturbances, and measurement uncertainties. The EL system is decomposed into two subsystems called the proxy subsystem and the virtual tracking subsystem. An adaptive safe controller based on barrier Lyapunov functions is designed for the virtual tracking subsystem to ensure the boundedness of the safe velocity tracking error, and a safe controller based on control barrier functions  is designed for the proxy subsystem to ensure controlled invariance of the safe set defined either in the joint space or task space. Theorems that guarantee the safety of the proposed controllers are provided. In contrast to existing safe control strategies for EL systems, the proposed method requires much less model information and can ensure safety rather than input-to-state safety.  Simulation results are provided to illustrate the effectiveness of the proposed method.
\end{abstract}

%---------------------------------------
%            Keywords                  |
%---------------------------------------

% \begin{keywords}
%   Function approximation technique, control barrier function, observers
% \end{keywords}

%%%%%%%%%%%%%%%%%%%%%%%%%%%%%%%%%%%%%%%%%%%%%%%%%%%%%%%%%%%%%%%%%%%%%%%%%%%%%%%%

%\IEEEpeerreviewmaketitle

%%%%%%%%%%%%%%%%%%%%%%%%%%%%%%%%%%%%%%%%%%%%%%%%%%%%%%%%%%%%%%%%%%%%%%%%%%%%%%%%

%---------------------------------------
%            MAIN PART                 |
%---------------------------------------
%=====================================================================
%        I   N   T   R   O   D   U   C   T   I   O   N               |
%=====================================================================
\section{Introduction}
\label{sec:introduction}
Safe-by-design control has received increasing interest because of its broad applications. Control Barrier Functions (CBFs) and Barrier Lyapunov Function (BLFs) are two widely investigated barrier type functions that can provably ensure \emph{safety} expressed as the controlled invariance of a given set \cite{ames2016control,Xu2015ADHS,jankovic2018robust,nguyen2021robust,tee2009barrier,ren2010adaptive,panagou2015distributed,jin2014barrier,salehi2020safe,wang2021observer,wang2022disturbance}. 
By integrating the CBF constraint into a convex quadratic program (QP), a CBF-QP-based controller is capable of serving as a safety filter that minimally alters possibly unsafe control inputs. In contrast, BLFs are Lyapunov-like functions defined in given open sets, such that they can ensure safety and stability simultaneously.

Euler-Lagrange (EL) systems, which represent a large number of mechanical systems including robot manipulators and vehicles, have been extensively investigated in the literature \cite{slotine1987adaptive,ortega2013passivity,cortez2022safe,capelli2022passivity}. Recently, the safe control of EL systems attracted significant attention because of the broad application of robotic systems in safety-critical scenarios, such as human-robot interaction. Many CBF-based control strategies have been developed for EL systems \cite{barbosa2020provably,ferraguti2022safety,farras2021safe}. Although these methods are demonstrated by both theoretical analysis and simulation/experimental results, they rely on model information of the EL system (i.e., the exact forms of the inertia matrix, the Coriolis/centripetal matrix, and the gravity term), which is hard to obtain precisely in practice. Few research has been devoted to the safe control of EL systems with limited model information \cite{singletary2021safety,molnar2021model}. In \cite{singletary2021safety}, a novel CBF that integrates kinetic energy with the classical form is proposed, resulting in reduced model dependence and less conservatism; however, this method does not take account of external disturbances, which are ubiquitous in practical applications. In \cite{molnar2021model}, a safe velocity is designed based on reduced-order kinematics and tracked by a velocity tracking controller; nevertheless, only \emph{input-to-state safety} \cite[Definition 3]{kolathaya2018input} rather than safety is ensured when the model information is unavailable, and the safe velocity is required to be differentiable. On the other hand, various BLF-based controllers have been developed for EL systems \cite{jin2014barrier,salehi2020safe}, whereas these approaches require the desired trajectory to stay inside the safe set and impose relatively strict structural requirements on safety constraints.

\begin{figure}[!t]
\centering
  \includegraphics[width=0.45\textwidth]{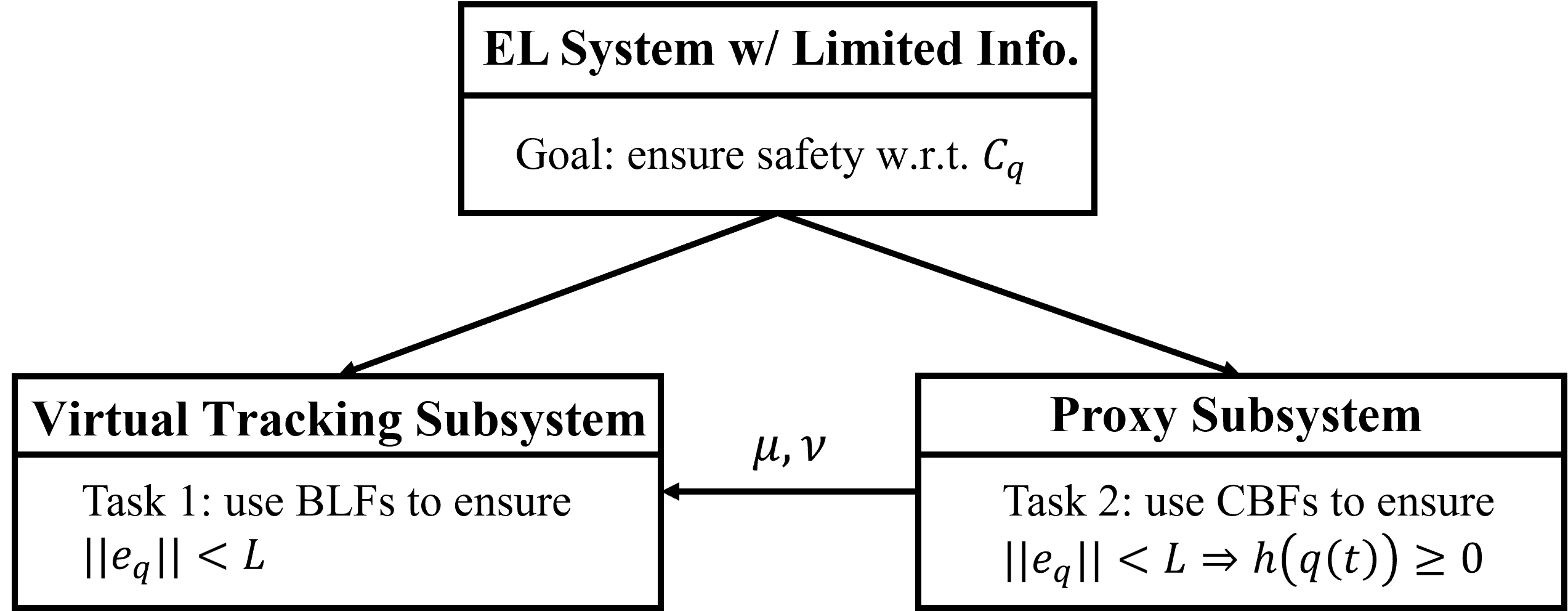}
\caption{Illustration of the proposed proxy-CBF-BLF control design scheme for safe control design of EL systems in the joint space. The original EL system is decomposed into the proxy subsystem and the virtual tracking subsystem. The safe velocity for the virtual tracking subsystem is generated by the proxy subsystem. A CBF-QP-based controller is designed for the proxy subsystem to ensure safety, while an adaptive BLF-based control law is proposed for the virtual tracking subsystem to constrain the safe velocity tracking error.}
\label{fig:taskillustraion}
\end{figure}

In this work, we propose a new control strategy for EL systems with limited model information, external disturbances, and measurement uncertainties. The original EL system is decomposed into two subsystems: the \textit{proxy subsystem}, which is a double integrator with a mismatched bounded disturbance, and the \textit{virtual tracking subsystem}, which corresponds to the dynamical model of the EL system. 
A CBF-based controller is designed for the proxy subsystem to generate the safe velocity, while an adaptive BLF-based controller is developed for the virtual tracking subsystem to track the safe velocity and ensure the boundedness of the tracking error.  See Fig. \ref{fig:taskillustraion} for illustration, where the symbols will be introduced in Section \ref{sec:joint}. 
Compared with existing results, the proposed method has four main advantages as shown in the following:
\begin{enumerate}
    \item The proposed method does not rely on any model information except for the upper bound of the inertia matrix's norm, which implies that even the bounds of the Coriolis-centrifugal and gravity matrices are not required in control design because such bounds are estimated by adaptive laws online.
    \item The closed-loop system is guaranteed to be safe, instead of input-to-state safe, in the presence of external disturbances and measurement uncertainties.
    \item The safe velocity's differentiability, which is important for velocity tracking control design, is guaranteed, and calculating its derivative is straightforward.
    \item The proposed method takes measurement uncertainties into account, allowing its use in robots where precise angular velocity measurements are not available.
\end{enumerate}

The remainder of this paper is organized as follows. In Section \ref{sec:preliminary},
preliminaries and the problem statement are introduced;
in Section \ref{sec:joint}, the joint space safe control strategy is presented; in Section \ref{sec:task}, the task space safe control scheme is shown;  in Section \ref{sec:simulation}, numerical simulation results are presented to validate the proposed method; and finally, the conclusion is drawn in Section \ref{sec:conclusion}.

\section{Preliminaries and Problem Statement}
\label{sec:preliminary}
Throughout the paper, we denote by $\R_{>0}$ and $\R_{\geq 0}$ the sets of positive real and nonnegative numbers, respectively. We denote $\|\cdot\|$ the 2-norm for vectors and the induced 2-norm  for matrices. We denote by $\sigma_{\rm min}(A)$ the smallest eigenvalue of a square matrix $A$. We consider the gradient $\frac{\pa h}{\pa x} \in \R^{n\times 1}$ as a row vector, where $x\in\R^n$ and $h:\R^n\to\R$ is a function with respect to $x$.

\subsection{Control Barrier Functions \& Barrier Lyapunov Functions}
CBFs and BLFs are two types of barrier functions that are widely used to ensure the controlled invariance of a given set \cite{ames2016control,tee2009barrier}. Our approach aims to combine the advantages of both CBFs and BLFs, which are briefly reviewed below.

\subsubsection{Control Barrier Functions}
Consider a control-affine system given as $\dot x = f(x)+g(x)u$ 
where $x\in\R^n$ is the state, $u\in U\subset \R^m$ is the control input, and $f: \mathbb{R}^n\to\mathbb{R}^n$ and $g:\mathbb{R}^n\to\mathbb{R}^{n\times m}$ are locally Lipchitz continuous functions. 
Define a \emph{safe set} $\mathcal{C}=\{ x \in \R^n: h(x) \geq 0\}$ 
where $h$ is a continuously differentiable function. 
The function $h$ is called a \emph{(zeroing) CBF} of relative degree 1, if  there exists a constant $\gamma>0$ such that $\sup_{u \in U}  \left[ L_f h(x) + L_g h(x) u + \gamma h(x)\right] \geq 0$  
where $L_fh(x)=\frac{\pa h}{\pa x} f(x)$ and $L_gh(x)=\frac{\pa h}{\pa x} g(x)$ are Lie derivatives \cite{Xu2015ADHS}. In this paper, we assume there is no constraint on the input $u$, i.e., $U=\R^m$. The following result that guarantees the forward invariance of $\mathcal{C}$ is given in \cite{Xu2015ADHS}.
\begin{lemma}\cite[Corollary 7]{Xu2015ADHS}
If $h$ is a (zeroing) CBF on $\R^n$, then any Lipschitz continuous controller $u: \R^n\to U$ such that $u(x) \in K(x)\triangleq \{ u\in U \mid L_f h(x) + L_g h(x) u + \gamma h(x) \geq 0\}$ will guarantee the forward invariance of $\mathcal{C}$, i.e., the \emph{safety} of the closed-loop system.     
\end{lemma}
By including the CBF condition into a convex QP, the provably safe controller is obtained by solving a CBF-QP online.  The time-varying CBF and its safety guarantee for a time-varying system are discussed in \cite{xu2018constrained}.

\subsubsection{Barrier Lyapunov Function}\label{sec:blf}
In contrast to CBFs, BLFs are positive definite functions that are more tightly connected with Lyapunov functions.
\begin{definition} \cite[Definition 2]{tee2009barrier}
A barrier Lyapunov function is a scalar function $V(x)$, defined
with respect to the system $\dot x = f (x)$ on an open region $\mathcal{D}$ containing the origin, that is continuous, positive definite, has
continuous first-order partial derivatives at every point of $\mathcal{D}$, has the property $V(x)\to\infty$ as $x$ approaches the boundary of $\mathcal{D}$, and satisfies $V(x(t))\leq b$ for any $t>0$ along the solution of $\dot x = f (x)$ for $x(0)\in\mathcal{D}$ and some positive constant $b$.
\end{definition}

The following lemma is used for BLF control design to guarantee that constraints on the output or state are satisfied. 

\begin{lemma}\label{lemmablf}\cite[Lemma 1]{tee2009barrier}
For any positive constants $k_{a_1}$, $k_{b_i}$, let $\mathcal{Z}_1\triangleq \{z_1\in\R: -k_{a_1}< z_1<k_{b_1}\}\subset \R$ and $\mathcal{N}\triangleq \R^l\times \mathcal{Z}_1\subset \R^{l+1}$ be open sets. Consider the system $\dot \eta = h(\eta,t)$ 
where $\eta\triangleq[w,z_1^\top]\in\mathcal{N}$, and $h:\R_{\geq 0}\times\mathcal{N}\to\R^{l+1}$ is  piecewise continuous in $t$ and locally Lipschitz in $z$, uniformly in $t$, on $\R_{\geq 0}\times\mathcal{N}$. Suppose that there exist functions $U:\R^l\to\R_{\geq 0}$ and $V_1:\mathcal{Z}_1\to\R_{\geq 0}$, continuously differentiable and positive definite in their respective domains, such that $V_1(z_1)\to\infty$ as $\ z_1\to-k_{a_1}$ or $z_1\to k_{b_1}$, and $\ga_1(\|w\|)\leq U(w)\leq\ga_2(\|w\|)$, 
where $\ga_1$ and $\ga_2$ are class $\mathcal{K}_\infty$ functions. Let $V(\eta)\triangleq V_1(z_1)+U(w)$, and $z_1(0)$ belong to the set $z_1\in(-k_{a_1},k_{b_1})$. If the inequality $\dot V\leq \frac{\pa V}{\pa \eta}h\leq 0$ holds, 
then $z_1(t)$ remains in the open set $z_1\in (-k_{a_1},k_{b_1})$, $\forall t\in[0,\infty)$.
\end{lemma}

\subsection{Euler-Lagrange Systems}
Consider an EL system given as follows \cite{ortega2013passivity,spong2006robot}:
\begin{IEEEeqnarray}{rCl}
\IEEEyesnumber\label{elsys:total}
\IEEEyessubnumber \label{elsys1}
       \dot q &=& \omega,\\
\IEEEyessubnumber\label{elsys2}
\dot \omega&=& M^{-1}(q) \left(\tau-C(q,\omega)\omega-G(q)+\tau_d\right),
\end{IEEEeqnarray}
where $q\in\R^n$ is the generalized coordinate, $\omega\in\R^n$ is the generalized velocity, $\tau\in\R^n$ is the control input, $\tau_d:\R_{\geq 0}\to\R^n$ is the external disturbance, $M:\R^n\to\R^{n\times n}$ is the inertia matrix, $C:\R^n\times\R^n\to\R^{n\times n}$ is the Coriolis/centripetal matrix, and $G:\R^n\to\R^n$ is the gravity term. We assume that the exact knowledge of the velocity $\omega$ is not known, and denote the measured generalized velocity as $\hat \omega$ (e.g., 
in some application scenarios, $\omega$ is obtained by numerically differentiating $q$ such that it may be contaminated by measurement noise); 
therefore the velocity measurement uncertainty can be defined as 
$$
\xi=\omega -\hat \omega.
$$ 
Furthermore, we assume that $\tau_d$, $\xi$, and $\dot\xi$ are all bounded. 
\begin{assumption}\label{assumptiontaud}
The disturbance $\tau_d$ satisfies $\|\tau_d\|\leq D_0$ where $D_0>0$ is a positive constant.
\end{assumption}

\begin{assumption}\label{assumptionnoise}
The measurement uncertainty $\xi$ and its derivative $\dot \xi$ are bounded as $\|\xi\|\leq D_1$ and $\|\dot\xi\|\leq D_2$, where $D_1$ and $D_2$ are positive constants.
\end{assumption}

Note that Assumption \ref{assumptiontaud} is extensively used in the robust control literature, and numerous state estimation techniques have been developed to ensure that the state estimation error is bounded.

The system given in \eqref{elsys:total} has the following properties that will be exploited in the subsequent control design \cite{dixon2007adaptive}.
\begin{property}[P1]\label{propertymass}
The matrix $M$ is positive definite,  symmetric, and satisfies 
\begin{align}\label{eqM}
\la_1\|q\|^2\leq q^\top M(q)q\leq\la_2\|q\|^2,\quad \forall q\in\R^n,    
\end{align}
where $\la_1$, $\la_2$ are positive constants.
\end{property}

\begin{property}[P2]\label{propertygravity}
The matrices $C(q,\omega)$ and $G(q)$ satisfy
\begin{align}
&\|C(q,\omega)\|\leq \zeta_c\|\omega\|, \quad
\|G(q)\|\leq \zeta_g,\quad\forall q,\omega\in\R^n,\label{eqCG}
\end{align}
where $\zeta_c$ and $\zeta_g$ are positive constants.
\end{property}

\subsection{Problem Statement}\label{sec:motivation}
In this work, we consider provably safe control design for an EL system given in \eqref{elsys:total} with \emph{limited information}. Specifically, we assume that the matrices $M,C,G$ in \eqref{elsys:total} are unknown and satisfy inequalities \eqref{eqM} and \eqref{eqCG} but only $\la_2$ is known. With such an EL system, the first problem we aim to solve is to design a feedback controller based on the knowledge of $q$ and $\hat \omega$ to ensure the safety of the system in the joint space.

\begin{problem}\label{probjoint}
Consider an EL system described by \eqref{elsys:total} where the matrices $M,C,G$ are unknown, and a joint space safe set $\C_q$ defined as
\begin{equation}\label{setcjoint}
    \C_q=\{q\in\R^n: h(q)\geq 0\},
\end{equation}
where $h$ is a twice differentiable function. Suppose that Assumptions \ref{assumptiontaud} and \ref{assumptionnoise} hold with $D_0$, $D_2$ unknown, and $M,C,G$ satisfy inequalities \eqref{eqM} and \eqref{eqCG} with constant $\la_2$ known  and constants $\la_1,\zeta_c,\zeta_g$ unknown. Design a feedback control law $\tau(q(t),\hat \omega(t),t)$ such that  the closed-loop system is always safe with respect to $\C_q$, i.e., $h(q(t))\geq 0,\forall t\geq 0$.
\end{problem}

The second problem we aim to solve is about designing a safe controller in the task space. 
\begin{problem}\label{probtask}
Consider an EL system described by \eqref{elsys:total} where the matrices $M,C,G$ are unknown, and the forward kinematics of the EL system:
\begin{equation}
    p = f(q),\label{forwardkine}
\end{equation}
where $p\in\R^k$ denotes the 
variable of the task space 
and $f:\R^n\to\R^k$  represents a continuously differentiable function with $k\leq n$. Consider a task space safe set $\C_p$ defined as
\begin{equation}\label{setctask}
    \C_p=\{p\in\R^p: h(p)\geq 0\},
\end{equation}
where $h$ is a twice differentiable function.  Suppose that Assumptions \ref{assumptiontaud} and \ref{assumptionnoise} hold with $D_0,D_2$ unknown, and $M,C,G$ satisfy inequalities \eqref{eqM} and \eqref{eqCG} with constant $\la_2$ known  and constants $\la_1,\zeta_c,\zeta_g$ unknown. Design a feedback control law $\tau(p(t), q(t),\hat \omega(t),t)$ such that  the closed-loop system is safe with respect to $\C_p$, i.e., $h(p(t))\geq 0,\forall t\geq 0$.
\end{problem}

The main difficulty of Problems \ref{probjoint} and \ref{probtask} lies in the limited information of the EL system: $\la_1$, $D_0$, $D_2$, $\zeta_c$, $\zeta_g$ are assumed to be unknown in control design. The proposed controller in this work is highly robust to model uncertainties and can be easily transferred between different EL systems without re-designing the control laws. Existing safe control design approaches for EL systems are not applicable to solve the problems in this work because they rely on the exact forms of $M$, $C$, $G$ or the values of  $\la_1$, $D_0$, $D_2$, $\zeta_c$, $\zeta_g$; see  \cite{cortez2022safe,capelli2022passivity,barbosa2020provably,ferraguti2022safety,farras2021safe,singletary2021safety,molnar2021model} for more details.

\section{Joint Space Safe Control}
\label{sec:joint}
In this section, a novel proxy-CBF-BLF-based method will be presented to solve Problem \ref{probjoint} for the EL system with limited information, external disturbances, and measurement uncertainties. We will show the main idea of the method in Subsection \ref{sub:overview}, propose an adaptive BLF-based control design approach for the virtual tracking subsystem in Subsection \ref{sub:jointblf}, and
a CBF-based control design strategy for the proxy subsystem in Subsection \ref{sec:jointcbf}. 

\subsection{Method Overview}\label{sub:overview}
The main idea of our method is to decompose an EL system into two subsystems, called the proxy\footnote{The term ``proxy'' is inspired by proxy-based sliding mode control \cite{kikuuwe2010proxy} and haptic rendering \cite{ruspini1997haptic}.} subsystem and the virtual tracking subsystem, and use the CBF and BLF to design safe controllers for the two subsystems, respectively, such that the overall controller will ensure the safety of the EL system (see Fig. \ref{fig:taskillustraion} for illustration). 

The proxy subsystem is given as:
\begin{IEEEeqnarray}{rCl}
\IEEEyesnumber\label{proxy}
\IEEEyessubnumber\label{proxy:1}
       \dot q&=& \mu + e_q+\xi,\\
\IEEEyessubnumber\label{proxy:2}
        \dot\mu&=& \nu,
\end{IEEEeqnarray}
where $\mu$ is the virtual safe velocity with $\mu(0)=\hat\omega(0)$, $e_q$ is the virtual velocity tracking error defined as
\begin{align}
   e_q=\hat \omega-\mu, 
\end{align} 
and $\nu$ is the virtual control input to be designed. Note that \eqref{proxy} is equivalent to \eqref{elsys1} augmented with an integrator. 

The virtual tracking subsystem is given as:
\begin{equation}\label{vts}
    \dot e_q=M(q)^{-1} (\tau-C(q,\omega)\omega -G(q)+\tau_d) - \dot\xi-\nu
\end{equation}
where $\tau$ is the control input to be designed and $\nu$ is from the proxy subsystem \eqref{proxy}.

With this decomposition, Problem \ref{probjoint} can be solved by accomplishing two tasks shown as follows.

\begin{task}\label{probblf}
For the virtual tracking subsystem \eqref{vts}, design a controller $\tau$ to guarantee 
\begin{align}\label{boundL}
\|e_q(t)\|<L,\forall t\geq 0,    
\end{align}
where $L>0$ is an arbitrary positive constant.
\end{task}

\begin{task}\label{probproxy system}
For the proxy subsystem \eqref{proxy}, design a control law $\nu$ to ensure $h(q(t))\geq 0,\forall t\geq 0$, under the assumption that $\|e_q(t)\|<L,\forall t\geq 0$.
\end{task}

\begin{remark}
In \cite{molnar2021model}, a safe velocity is designed based on reduced-order kinematics, which is similar to \eqref{proxy:1} in our proxy subsystem. However, including an additional integrator as shown in \eqref{proxy:2} is important because $\dot\mu$, which is equal to $\nu$,  is required in the virtual tracking subsystem \eqref{vts} and $L$ can be selected to be arbitrarily small, thereby reducing the potential conservatism of the safe controller (see Remark \ref{remark:BLF}).
Nevertheless, the added integrator will result in a system with a mismatched virtual disturbance, $e_q+\xi$; a new CBF-based safe control scheme will be proposed for such a system in Section \ref{sec:jointcbf}.
\end{remark}

\subsection{BLF-based Control For the Virtual Tracking Subsystem}\label{sub:jointblf}
In this subsection, an adaptive BLF-based controller will be presented to accomplish Task \ref{probblf}. The BLF-based method is suitable for this task because
it does not rely on the bounds of the unknown parameters and the external disturbances.

Inspired by our previous work \cite{wang2023adaptive}, the following theorem presents a controller $\tau$ for the virtual tracking subsystem to ensure $\|e_q(t)\|<L,\forall t\geq 0$. 

\begin{theorem}\label{theorem:blf}
Consider the virtual tracking subsystem \eqref{vts} where the matrices $M,C,G$ are unknown. Suppose that Assumptions \ref{assumptiontaud} and \ref{assumptionnoise} hold with $D_0$, $D_2$ unknown, and $M,C,G$ satisfy inequalities \eqref{eqM} and \eqref{eqCG} with constant $\la_2$ known  and constants $\la_1,\zeta_c,\zeta_g$ unknown. Suppose that the controller $\tau$ is designed as
\begin{align}
\tau &=-\la_2e_q\mathcal{N}\label{blfu}
\end{align}
where 
\begin{IEEEeqnarray}{rCl}
\IEEEyesnumber\label{blfcontrol}
\IEEEyessubnumber\label{blfadaptive0}\hspace{-4mm}
\mathcal{N}\!&=&\!k_1\!+\! 
\frac{(\hat\ta_1\varphi)^2}{\hat\ta_1\varphi\|e_q\|\!+\!\ep_1}\!+\!\frac{\hat\ta_2^2}{\hat\ta_2\|e_q\|\!+\!\ep_2}\!+\!\frac{\|\nu\|^2}{\|e_q\|\|\nu\|\!+\!\ep},\\
\IEEEyessubnumber\label{blfadaptive1}       \hspace{-4mm}\dot{\hat{\ta}}_1&=& -\ga_\ta\hat\ta_1+\frac{\|e_q\|\varphi}{L^2-\|e_q\|^2},\\
       \IEEEyessubnumber\label{blfadaptive2}
\hspace{-4mm}\dot{\hat{\ta}}_2&=& -\ga_\ta\hat\ta_2+\frac{\|e_q\|}{L^2-\|e_q\|^2}, 
\end{IEEEeqnarray}
and $\varphi=(\|\hat \omega\|+D_1)^2$, 
with positive constants $\ep, \ep_1,\ep_2, \ga_\ta>0$, $k_1>\frac{\Lambda}{L^2}$, and $\Lambda=\ep+\ep_1+\ep_2$.
If $\hat\ta_1(0),\hat\ta_2(0)> 0$, then $\|e_q(t)\|<L$ for any $t\geq 0$. 
\end{theorem}
\begin{proof}
From \eqref{blfadaptive1}-\eqref{blfadaptive2}, $\dot {\hat{\ta}}_{1} \geq -\gamma_\ta \hat\ta_1,\dot {\hat{\ta}}_{1} \geq -\gamma_\ta \hat\ta_2$ hold in the open set $\mathcal{Z}_L\triangleq\{e_q\in\R^n\mid \|e_q\|<L\}$. Since $\hat\ta_{1}(0)>0,\hat\ta_{2}(0)>0$, it is easy to see that $\hat\ta_{1}(t)\geq 0$ and $\hat\ta_{2}(t)\geq 0$ for any $t\geq 0$ by the Comparison Lemma \cite[Lemma 2.5]{khalil2002nonlinear}.

Define
$\ta_1=\zeta_c\la_1^{-1}$ and $\ta_2=\la_1^{-1}(\zeta_g+D_0)+D_2$, which are unknown parameters because $\la_1,\zeta_c,\zeta_g, D_2$ are unknown. Define a candidate BLF as
\begin{equation}\label{blf2}
    V=\fot\log\left(\frac{L^2}{L^2-\| e_q\|^2}\right)+ \fot\tilde{\ta}_1^2+\fot\tilde{\ta}_2^2,
\end{equation}
where $\tilde\ta_1=\ta_1-\hat\ta_1$, $\tilde\ta_2=\ta_2-\hat\ta_2$. The derivative of $V$ in the open set $\mathcal{Z}_L$ can be expressed as 
\begin{IEEEeqnarray}{rCl}
       \dot V&=&\frac{e_q^\top}{L^2-\| e_q\|^2}(M^{-1} (\tau-C(q, \omega)\omega -G(q)+\tau_d)\nonumber\\
       &&- \dot\xi-\nu)- \tilde\ta_1\dot{\hat{\ta}}_1- \tilde\ta_2\dot{\hat{\ta}}_2 \nonumber\\
       &\leq& \frac{e_q^\top M^{-1} \tau}{L^2-\| e_q\|^2}+\frac{\|e_q\|}{L^2-\|e_q\|^2}(\|M^{-1}\|(\|C(q, \omega) \omega\|\nonumber\\
       &&+\|G\|+\|\tau_d\|)+\|\dot\xi\|+\|\nu\|)- \tilde\ta_1\dot{\hat{\ta}}_1- \tilde\ta_2\dot{\hat{\ta}}_2\nonumber\\
       &\leq& \frac{e_q^\top M^{-1} \tau}{L^2-\|e_q\|^2}+\frac{\|e_q\|}{L^2-\|e_q\|^2}(\la_1^{-1}(\zeta_c(\|\hat \omega\|+D_1)^2\nonumber\\
       &&+\zeta_g+D_0)+D_2+\|\nu\|)- \tilde\ta_1\dot{\hat{\ta}}_1- \tilde\ta_2\dot{\hat{\ta}}_2 \nonumber\\
    %   &=&\frac{e^\top M^{-1} \tau}{L^2-e^\top e}+\frac{\|e\|}{L^2-e^\top e}(\ta_1\varphi_1+\ta_2+\|\nu\|)- \tilde\ta_1\dot{\hat{\ta}}_1- \tilde\ta_2\dot{\hat{\ta}}_2\nonumber\\
       &=&\frac{e_q^\top M^{-1} \tau}{L^2-\|e_q\|^2}+\frac{\|e_q\|}{L^2-\|e_q\|^2}(\hat\ta_1\varphi+\hat\ta_2+\|\nu\|)\nonumber\\
       &&-\tilde\ta_1\!\left(\!\dot{\hat{\ta}}_1\!-\!\frac{\|e_q\|\varphi}{L^2\!-\!\|e_q\|^2}\!\right)\!-\!\tilde\ta_2\!\left(\!\dot{\hat{\ta}}_2\!-\!\frac{\|e_q\|}{L^2\!-\!\|e_q\|^2}\!\right),\label{dotv1}
\end{IEEEeqnarray}
where the second inequality comes from 
\begin{equation*}
    \|C(q,\omega)\omega\|\stackrel{{\rm (P1)}}{\leq} \zeta_c\|\omega\|^2=\zeta_c\|\hat \omega+\xi\|^2\leq \zeta_c(\|\hat \omega\|+D_1)^2,
\end{equation*}
and the third inequality arises from the fact $\la_2^{-1}\leq\|M(q)^{-1}\|\leq \la_1^{-1}$ for any $q\in\R^n$, according to Property \ref{propertymass}. 
Substituting \eqref{blfcontrol} into \eqref{dotv1} yields
\begin{IEEEeqnarray}{rCl}
       \dot V&\leq&\frac{1}{L^2\!-\!\|e_q\|^2}\bigg(\!
       -\la_2\underbrace{(e_q^\top\! M^{-1}e_q)}_{\geq \la_2^{-1}\|e_q\|^2}
       \bigg(\! k_1\!+\!\underbrace{
       \frac{(\hat\ta_1\varphi)^2}{\hat\ta_1\varphi\|e_q\|\!+\!\ep_1}}_{\geq 0}
       \nonumber\\
       &&\!+\!\underbrace{\frac{(\hat\ta_2)^2}{\hat\ta_2\|e_q\|\!+\!\ep_2}}_{\geq 0}\!+\!\underbrace{\frac{\|\nu\|^2}{\|e_q\|\|\nu\|\!+\!\ep}}_{\geq 0}
       \bigg)\!+\! \|e_q\|(\hat\ta_1\varphi\!+\!\hat\ta_2\!+\!\|\nu\|)\bigg)\nonumber\\
       &&+\ga_\ta(\tilde\ta_1\hat\ta_1+\tilde\ta_2\hat\ta_2)\nonumber\\
       &\leq&\frac{1}{L^2-\|e_q\|^2}\bigg(
       -k_1\|e_q\|^2+\left(
       \hat\ta_1\varphi\|e_q\|-\frac{(\hat\ta_1\varphi\|e_q\|)^2}{\hat\ta_1\varphi\|e_q\|+\ep_1}\right)\nonumber\\
       &&+\left(\!\hat\ta_2\|e_q\|\!-\!\frac{(\hat\ta_2\|e_q\|)^2}{\hat\ta_2 \|e_q\|\!+\!\ep_2}\right)\!+\!\left(\!\|e_q\|\|\nu\|\!-\!\frac{(\|e_q\|\|\nu\|)^2}{\|e_q\|\|\nu\|\!+\!\ep}\!\right)
       \bigg)\nonumber\\
       &&+\ga_\ta(\tilde\ta_1\hat\ta_1+\tilde\ta_2\hat\ta_2)\nonumber\\
       &\leq& \frac{1}{L^2-\|e_q\|^2}\left(-k_1\|e_q\|^2+\Lambda\right)+\ga_\ta(\tilde\ta_1\hat\ta_1+\tilde\ta_2\hat\ta_2),\label{dotv2}
\end{IEEEeqnarray}
where the last inequality comes from the fact that for any $A\geq 0,\ep> 0$,
$A-\frac{A^2}{A+\ep}=\frac{A\ep}{A+\ep}\leq \ep$ holds true. Noting that $\frac{\Lambda}{L^2-\|e_q\|^2}=\frac{\Lambda}{L^2}+\frac{\Lambda}{L^2}\frac{\|e_q\|^2}{L^2-\|e_q\|^2}$ and $\tilde\ta_i\hat\ta_i=\tilde\ta_i(\ta_i-\tilde\ta_i)\leq\frac{\ta_i^2-\tilde\ta_i^2}{2},\ i=1,2$, we have
\begin{IEEEeqnarray*}{rCl}
       \dot V&\leq& -\frac{\chi\|e_q\|^2}{L^2-\|e_q\|^2}-\frac{\ga_\ta}{2}\sum_{i=1}^2\tilde{\ta}_i^2+\frac{\Lambda}{L^2}+\frac{\ga_\ta}{2}\sum_{i=1}^2\ta_i^2\nonumber\\
       &\leq& -\chi\log\left(
       \frac{L^2}{L^2-\|e_q\|^2}
       \right)-\frac{\ga_\ta}{2}\sum_{i=1}^2\tilde{\ta}_i^2+\frac{\Lambda}{L^2}+\frac{\ga_\ta}{2}\sum_{i=1}^2\ta_i^2\nonumber\\
       &\leq& -\kappa V+K,
\end{IEEEeqnarray*}
where $\chi=k_1-\frac{\Lambda}{L^2}$,
$K=\frac{\Lambda}{L^2}+\frac{\ga_\ta}{2}\sum_{i=1}^2\ta_i^2$, $\kappa=\min\{2\chi,\ga_\theta\}$, and the second inequality comes from the fact that $\log\frac{L^2}{L^2-\|e_q\|^2}\leq \frac{\|e_q\|^2}{L^2-\|e_q\|^2}$ holds in the open set $\mathcal{Z}_L$ \cite[Lemma 2]{ren2010adaptive}. Thus, $V(t)$ is bounded, which implies  $\|e_q(t)\|<L$ for any $t\geq 0$, according to Lemma \ref{lemmablf}.
\end{proof}

\subsection{CBF-based Control For the Proxy Subsystem}
\label{sec:jointcbf}
In this subsection, a CBF-based control law is presented to solve Task \ref{probproxy system}. Note that designing a CBF-based controller for accomplishing Task \ref{probproxy system} is challenging because the term $e_q+\xi$ is considered as a bounded mismatched disturbance to the proxy subsystem and its derivative, $\dot e_q+\dot\xi$, is not necessarily bounded.

Since in Task \ref{probproxy system} we assume $\|e_q\|<L$ holds, which is ensured by Theorem \ref{theorem:blf}, the term $e_q+\xi$ is bounded as
\begin{equation}
    \|e_q+\xi\|< D_1+L. \label{exi}
\end{equation}

The following theorem provides a CBF-based controller $\nu$ that ensure  $h(q(t))\geq 0,\forall t\geq 0$.

\begin{theorem}\label{theorem:cbf}
Consider the proxy subsystem given in \eqref{proxy} and a joint space safe set $\C_q$ defined in \eqref{setcjoint}. Suppose that $h(q(0))> 0$, $\|e_q(t)\|<L$ for any $t\geq 0$, and there exist positive constants $\la,\ga,\beta>0$ such that \\
(i) $\frac{\pa h}{\pa q}(q(0))\mu(0)-\frac{1}{2\beta}\left\|\frac{\pa h}{\pa q}(q(0))\right\|^2-\frac{\beta(D_1+L)^2}{2}+\la h(q(0))\geq 0$;\\
(ii) the set $K_{BF}^q(q,\mu)=\{{\mathfrak u}  \in\R^n: \Psi_0+\Psi_1 {\mathfrak u}\geq 0\}$ is not empty for any $q\in\C_q$ and $\mu\in\R^n$, where
\begin{IEEEeqnarray}{rCl}
       \IEEEyesnumber \label{psi}
       \IEEEyessubnumber\label{psi0}
       \Psi_0&=& \mathcal{M}\mu- \left\|\mathcal{M}\right\|(D_1+L)+\ga\bar h,\\
       \IEEEyessubnumber\label{psi1}
       \Psi_1&=& \frac{\pa h}{\pa q},
\end{IEEEeqnarray}
with $\mathcal{M}=\mu^\top \rh-\frac{1}{\beta}\frac{\pa h}{\pa q}\rh+\la
       \frac{\pa h}{\pa q}$,
$\rh=\frac{\pa^2 h}{\pa q^2}$ denotes the Hessian, and   
       $\bar h = \frac{\pa h}{\pa q}\mu-\frac{1}{2\beta}\left\|\frac{\pa h}{\pa q}\right\|^2-\frac{\beta(D_1+L)^2}{2}+\la h$.\\
Then, any Lipschitz continuous control input $\nu\in K_{BF}^q(q,\mu)$ will make $h(q(t))\geq 0$ for any $t\geq 0$.
\end{theorem}
\begin{proof}
First, we show that $\nu\in K_{BF}^q(q,\mu)\implies\bar h(t)\geq 0$ for any $t\geq 0$. Note that Condition (i) indicates that $\bar h(q(0),\mu(0))\geq 0$. Meanwhile, one can observe that $\dot{\bar{h}}$ can be expressed as
\begin{IEEEeqnarray*}{rCl}
       \dot{\bar{h}}&=& \frac{\pa h}{\pa q}\nu+\left(\mu^\top \rh-\frac{1}{\beta}\frac{\pa h}{\pa q}{\rm H}_h+\la
       \frac{\pa h}{\pa q}\right)(\mu+e_q+\xi)\nonumber\\
       &=&\frac{\pa h}{\pa q}
       \nu+ \mathcal{M}\mu+\mathcal{M}(e+\xi)\nonumber\\
       &\stackrel{\eqref{exi}}{\geq}& \Psi_1 \nu + \mathcal{M}\mu- \left\|\mathcal{M}\right\|(D_1+L).
\end{IEEEeqnarray*}
Selecting $\nu\in K_{BF}^q(q,\mu)$ yields $ \dot{\bar{h}}\geq -\Psi_0+\mathcal{M}\mu- \left\|\mathcal{M}\right\|(D_1+L)=-\ga\bar h,$ 
which indicates $\bar h(q(t),\mu(t))\geq 0,\forall t\geq 0$ because  $\bar h(q(0),\mu(0))\geq 0$. 
Since
\begin{IEEEeqnarray*}{rCl}
       \dot h+\la h&=& \frac{\pa h}{\pa q}(\mu+e_q+\xi)+\la h\nonumber\\
       &\geq& \frac{\pa h}{\pa q}\mu-\frac{1}{2\beta}\left\|\frac{\pa h}{\pa q}\right\|^2-\frac{\beta}{2}\|e_q+\xi\|^2+\la h\nonumber\\
       &\stackrel{\eqref{exi}}{\geq}&\frac{\pa h}{\pa q}\mu-\frac{1}{2\beta}\left\|\frac{\pa h}{\pa q}\right\|^2-\frac{\beta(D_1+L)^2}{2}+\la h\nonumber\\
       &=&\bar h(q,\mu)\geq 0,
\end{IEEEeqnarray*}
one can conclude that 
$h(q(t))\geq 0,\forall t$ since $h(q(0))\geq 0$. 
\end{proof}

The safe virtual controller proposed in Theorem \ref{theorem:cbf} is obtained by solving the following CBF-QP:%
\begin{align}
\min_{\nu} \quad & \|\nu-\nu_{d}\|^2\label{cbfQP1}\\
\textrm{s.t.} \quad & \Psi_0+\Psi_1 \nu\geq 0, \nonumber
\end{align}
where $\Psi_0,\Psi_1$ are given in \eqref{psi} and $\nu_{d}$ is any given nominal control law.

The safe feedback control law $\tau(q(t),\hat \omega(t),t)$ to the EL system \eqref{elsys:total} consists of the control law $\tau$ given in \eqref{blfu} and the control law $\nu$ given in \eqref{cbfQP1}. By Theorems \ref{theorem:blf} and \ref{theorem:cbf}, the safe controller will ensure that the closed-loop system is always safe with respect to $\C_q$, i.e., $h(q(t))\geq 0$ for all $t\geq 0$.

\begin{remark}
The nominal control law $\nu_d$ can be designed as $\nu_d=-\alpha_1E_{q}-\alpha_2E_{\mu}+\ddot q_d$, where $E_{q}\triangleq q-q_{d}$, $E_{\mu}\triangleq \mu-\dot q_{d}$, $q_{d}$ denotes the reference trajectory, and
$\alpha_1,\alpha_2\in\R$ are selected such that
\begin{equation*}
    \sigma_{\rm min}\left(\begin{bmatrix}
    0_{n\times n}&I_{n\times n}\\ -\alpha_1I_{n\times n}&-\alpha_2I_{n\times n}
    \end{bmatrix}\right)\triangleq -\alpha<-\fot.
\end{equation*}
Define a Lyapunov candidate function as $V=\fot \varepsilon^\top \varepsilon$, where $\varepsilon=[E_q^\top\ E_\mu^\top]^\top$. Since $\dot V$ satisfies
\begin{IEEEeqnarray*}{rCl}
    \dot V
    &=&\varepsilon^\top
    \begin{bmatrix}
    0_{n\times n}&I_{n\times n}\\ -\alpha_1I_{n\times n}&-\alpha_2I_{n\times n}
    \end{bmatrix}\varepsilon+E_{q}^\top (e_q+\xi)\nonumber\\
    &\leq&-\left(2\alpha-1\right)V+\frac{(D_1+L)^2}{2},
\end{IEEEeqnarray*}
the tracking error is uniformly ultimately bounded \cite{khalil2002nonlinear}.
\end{remark}

\begin{remark}
Suppose that $q^*$ is the unique zero of $\Psi_1$ in $\C_q$. We claim that if 
\begin{equation}\label{cbfconditionspecial}
    \rh^*\triangleq\rh(q^*)\succ aI_{n\times n},\
    h^*\triangleq h(q^*)>0,
\end{equation}
where $a$ is an arbitrary positive constant, then one can always find $\ga,\beta,\la>0$ such that Condition (i) and (ii) in Theorem \ref{theorem:cbf} hold true. Indeed, one can easily select $\beta$ and $\la$ such that Condition (i) is fulfilled. Meanwhile,
from \eqref{psi} one can observe that $\Psi_0^*\triangleq\Psi_0(q^*,\mu)=\mu^\top \rh^* \mu-\|\mu^\top \rh^*\|(D_1+L)+\ga\la h-\frac{\beta(D_1+L)^2}{2}$ satisfies
\begin{IEEEeqnarray*}{rCl}
    \Psi_0^*&\geq& a\|\mu\|^2\!-\!\|\rh^*\|(D_1\!+\!L)\|\mu\|\!+\!\ga\la h^*\!-\!\frac{\beta(D_1\!+\!L)^2}{2}\nonumber\\
    &=&a\left(\|\mu\|-\frac{\|\rh^*\|(D_1+L)}{2a}\right)^2+\ga\la h^*-\Xi\nonumber\\
    &\geq& \ga\la h^*-\Xi,
\end{IEEEeqnarray*}
where $\Xi=\frac{\|\rh^*\|^2(D_1+L)^2}{4a}+\frac{\beta(D_1+L)^2}{2}$. It is obvious that selecting $\ga\geq\frac{\Xi}{\la h^*}$ will yield $\Psi_0\geq 0$, such that $K_{BF}^q$ is not empty when $q=q^*$, which
shows the correctness of the claim. Furthermore, it is obvious that if $\Psi_1$ has finite zeros in $\C_q$ and each zero satisfies \eqref{cbfconditionspecial}, then one can always select appropriate $\ga,\la,\beta$ such that Conditions (i) and (ii) of Theorem \ref{theorem:cbf} are satisfied. Nevertheless, it should be noticed that \eqref{cbfconditionspecial} is not the unique criterion for verifying the conditions in Theorem \ref{theorem:cbf}. Developing systematic methods to design $h$ satisfying these conditions will be our future work.
\end{remark}

\begin{remark}\label{remark:BLF}
The bound $L$ for $\|e_q\|$ as given in \eqref{boundL}  should be carefully selected to achieve a trade-off between the control performance and the maximum magnitude of the control input.  If $L$ is selected to be very small, the control input tends to be significant because the state is more likely to approach the boundary of the output constraint; if $L$ is chosen to be large, unnecessary conservatism (i.e., the system only operates in a subset of the original safety set) may be introduced because in Theorem \ref{theorem:cbf} the worst-case of $e_q+\xi$ is considered.

If the proxy subsystem is not augmented with an additional integrator, the requirement of $L$ would be more restrictive, i.e., $L\geq |\hat\omega(0)-\mu(q(0))|$ is required. 
In practice, this may necessitate the selection of a larger $L$, which could result in unnecessarily conservatism. Meanwhile, $L$ is used in the design of $\mu$ to guarantee safety, which implies that $\mu(q(0))$ implicitly relies on $L$. Thus, in some cases, it may be difficult to find an appropriate $L$ that satisfies $L\geq |\hat\omega(0)-\mu(q(0))|$.
\end{remark}

\section{task space Safe Control}
\label{sec:task}

In this section, we will utilize the idea presented in the preceding section to solve the task space safe control problem for the EL system with limited information, external disturbances, and measurement uncertainties. The proxy subsystem in task space is more complicated to control than that in joint space; therefore, a different CBF-based control scheme is proposed.

Invoking \eqref{forwardkine}, one can see
\begin{equation}\label{jacobian}
    \dot p = J(q)\omega = J(q)(\hat \omega+\xi),
\end{equation}
where $J=\frac{\pa f}{\pa q}$ denotes the Jacobian \cite{spong2006robot}. Substituting \eqref{jacobian} into \eqref{elsys:total} yields
\begin{IEEEeqnarray}{rCl}
\IEEEyesnumber\label{elsysp:total}
\IEEEyessubnumber \label{elsysp1}
       \dot p &=& J(q)(\hat \omega+\xi),\\
\IEEEyessubnumber\label{elsysp2}
\dot{\hat{\omega}}&=& M^{-1}(q) (\tau-C(q,\omega)\omega-G(q)+\tau_d)-\dot\xi.
\end{IEEEeqnarray}
System \eqref{elsysp:total} can be decomposed into the proxy subsystem and the virtual tracking subsystem similar to Section \ref{sec:joint}. The proxy subsystem is given as:
\begin{IEEEeqnarray}{rCl}
\IEEEyesnumber\label{proxy systemp}
\IEEEyessubnumber\label{proxy systemp:1}
       \dot p&=& J(q)\eta + J(q)(e_p+\xi),\\
\IEEEyessubnumber\label{proxy systemp:2}
        \dot\eta&=& \upsilon,
\end{IEEEeqnarray}
where $\eta$ is the virtual state with $\eta(0)=\hat \omega(0)$, $e_p\triangleq\hat \omega-\eta$, and $\upsilon$ denotes the virtual control input to be designed. The virtual tracking subsystem is given as:
\begin{equation}\label{vtsp}
    \dot e_p=M^{-1} (\tau-C(q,\omega)\omega -G(q)+\tau_d) - \dot\xi-\upsilon.
\end{equation}
Note that system \eqref{vtsp} corresponds to system \eqref{vts}, for which the adaptive BLF-based controller developed in Theorem \ref{theorem:blf} is still applicable.
On the other hand, the CBF-based controller presented in Theorem \ref{theorem:cbf} is inapplicable to the proxy subsystem given in \eqref{proxy systemp} because \eqref{proxy systemp} is different from \eqref{proxy}. 
We will design a new CBF-based safe control law for \eqref{proxy systemp} to ensure the forward invariance of $\C_p$. To that end, we first design 
a nominal tracking controller for the proxy subsystem \eqref{proxy systemp} based on backstepping \cite{kokotovic1992joy} as shown in  the following proposition.
\begin{proposition}\label{theorem:tasktracking}
Consider the proxy subsystem \eqref{proxy systemp} and a desired trajectory $p_d$. Suppose that $\|e_p\|< L$ and the Jacobian $J$ has full row rank, i.e., there exists $J^\dagger$ such that $J J^\dagger=I_{k\times k}$. If the desired control input $\upsilon_d$ is designed as
\begin{IEEEeqnarray}{rCl}
\IEEEyesnumber\label{vdtracking}
\IEEEyessubnumber\label{eta}
\delta&=&J^\dagger\left(-l_1\ep_d+\dot p_d-\frac{\|J\|^2}{2}\ep_d\right) ,\\
\IEEEyessubnumber\label{nud}
    \upsilon_d&=& -l_2 \ep_\eta\!+\!\frac{\pa\delta}{\pa p}J\eta\!+\!\frac{\pa\delta}{\pa t}\!-\!\fot\left\|\frac{\pa \delta}{\pa p}J\right\|^2\ep_\eta\!-\!J^\top \ep_d,
\end{IEEEeqnarray}
where $\ep_d=p-p_d$, $\ep_\eta=\eta-\delta$, and $l_1,l_2>0$ are arbitrary positive constants,
then the tracking error $\ep_d$ is uniformly ultimately bounded.
\end{proposition}
\begin{proof}
Define a Lyapunov candidate function as $V_1=\fot \ep_d^\top \ep_d$. The derivative of $V_1$ satisfies
\begin{IEEEeqnarray*}{rCl}
\dot V_1&=& \ep_d^\top (J\delta+J\ep_\eta+J(e_p+\xi)-\dot p_d)\nonumber\\
&\leq& \ep_d^\top (J\delta+J\ep_\eta -\dot p_d)+\frac{\|\ep_d\|^2\|J\|^2}{2}+\frac{(D_1+L)^2}{2}\nonumber\\
&\stackrel{\eqref{eta}}{\leq}&-l_1\|\ep_d\|^2+\ep_d^\top J\ep_\eta+\frac{(D_1+L)^2}{2}.
\end{IEEEeqnarray*}
Then, an augmented Lyapunov candidate function is designed as $V_2=V_1+\fot \ep_\eta^\top \ep_\eta$, whose derivative can be expressed as
\begin{IEEEeqnarray*}{rCl}
\dot V_2&\leq &-l_1\|\ep_d\|^2+\ep_\eta^\top J^\top \ep_d+\frac{(D_1+L)^2}{2}\nonumber\\
&&+ \ep_\eta^\top\left(\upsilon_d-\frac{\pa \delta}{\pa p}J(\eta+e_p+\xi)-\frac{\pa \delta}{\pa t}\right)\nonumber\\
&\leq&-l_1\|\ep_d\|^2+\ep_\eta^\top J^\top \ep_d+(D_1+L)^2\nonumber\\
&&+ \ep_\eta^\top\left(\upsilon_d-\frac{\pa \delta}{\pa p}J\eta-\frac{\pa \delta}{\pa t}\right)+\fot\|\ep_\eta\|^2\left\|\frac{\pa\delta}{\pa p}J\right\|^2\nonumber\\
&\stackrel{\eqref{nud}}{\leq}&-l_1\|\ep_d\|^2-l_2\|\ep_\eta\|^2+(D_1+L)^2.%\label{dotv2p}
\end{IEEEeqnarray*}
Therefore, the tracking error $\ep_d$ is uniformly ultimately bounded \cite{khalil2002nonlinear}.
\end{proof}

A CBF-based safe control law is proposed for the proxy subsystem \eqref{proxy systemp} in the following theorem. 
\begin{theorem}\label{theorem:task}
Consider the proxy subsystem \eqref{proxy systemp} and the set $\C_p$ defined in \eqref{setctask}. Suppose that $h(p(0))\geq 0$, $\|e_p(t)\|< L,\forall t\geq 0$, and there exist constants $\la,\ga,\beta>0$ such that \\
(i) $\frac{\pa h}{\pa p}(p(0))J(q(0))\eta(0)-\frac{1}{2\beta}\left\|\frac{\pa h}{\pa p}(p(0))J(q(0))\right\|^2-\frac{\beta(D_1+L)^2}{2}+\la h(p(0))\geq 0$;\\
(ii) the set $K_{BF}^p(p,q,\mu)=\{{\mathfrak u}  \in\R^n: \Psi_0+\Psi_1 {\mathfrak u}\geq 0\}$ is not empty for any $p\in\C_p$ and $\eta\in\R^n$, where
\begin{IEEEeqnarray}{rCl}
       \IEEEyesnumber \label{psitask}
       \Phi_0&=& \frac{\pa\bar h}{\pa q}(J\eta+\hat \omega)-\left\|\frac{\pa\bar h}{\pa q}J\right\|(D_1+L)\nonumber\\
       \IEEEyessubnumber
       &&-\left\|\frac{\pa\bar h}{\pa q}\right\|D_1+\ga\bar h,\label{psitask0}\\
       \IEEEyessubnumber\label{psitask1}
       \Phi_1&=& \frac{\pa h}{\pa p}J,
\end{IEEEeqnarray}
with $\bar h=\frac{\pa h}{\pa p}J\eta-\frac{1}{2\beta}\left\|\frac{\pa h}{\pa p}J\right\|^2-\frac{\beta(D_1+L)^2}{2}+\la h$.\\
Then, any Lipschitz continuous control input $\upsilon\in K_{BF}^p$ will make $h(p(t))\geq 0$ for any $t\geq 0$.
\end{theorem}
\begin{proof}
We only show the sketch of the proof due to space limitation and the similarity of the proof to that of Theorem \ref{theorem:cbf}. One can see that selecting $\upsilon\in K_{BF}^p$ ensures $\dot{\bar{h}}\geq -\ga\bar h$; therefore,  $\bar h(t)\geq 0$ for any $t\geq 0$ since Condition (i) implies $\bar h(p(0),q(0),\eta(0))\geq 0$. Then, it can be proved that $\bar h(t)\geq 0\implies h(t)\geq 0$ for any $t\geq 0$. 
\end{proof}

Based on Proposition \ref{theorem:tasktracking} and Theorem \ref{theorem:task}, 
the safe virtual controller $\nu$ can be obtained by solving a CBF-QP:%
\begin{align}
\min_{\nu} \quad & \|\upsilon-\upsilon_{d}\|^2\label{cbfQP2}\\
\textrm{s.t.} \quad & \Phi_0+\Phi_1 \upsilon\geq 0, \nonumber
\end{align}
where $\Phi_0,\Phi_1$ are given in \eqref{psitask} and $\upsilon_{d}$ is presented in \eqref{vdtracking}.

The safe feedback control law $\tau(p(t), q(t),\hat \omega(t),t)$ to the EL system \eqref{elsys:total} consists of the control law $\tau$ given in \eqref{blfu} and the control law $\nu$ given in \eqref{cbfQP2}. By Theorems \ref{theorem:blf} and \ref{theorem:task}, the control law $\tau(p(t), q(t),\hat \omega(t),t)$ will ensure the safety of the closed-loop system with respect to $\C_p$, i.e., $h(p(t))\geq 0$ for all $t\geq 0$.

\section{Simulation}
\label{sec:simulation}
In this section, numerical simulation results are presented to demonstrate the effectiveness of the proposed method. Consider a two-linked robot manipulator, whose dynamics can be described by \eqref{elsys:total} with 
\begin{equation}
M(q)\!=\!
\begin{bmatrix}
\frac{m_1l^2}{3}\!+\!\frac{4m_2l^2}{3}\!+\!m_2 l^2\cos q_2& \frac{m_2l^2}{3}\!+\!\frac{m_2l^2}{2}\cos q_2 \\
\frac{m_2l^2}{3}\!+\!\frac{m_2l^2}{2}\cos q_2 & \frac{m_2 l^2}{3}
\end{bmatrix},\nonumber
\end{equation}
\begin{equation}
C(q, \omega)\!=\!
\begin{bmatrix}
-\frac{m_2l^2}{2}\dot q_2\sin q_2 &-\frac{m_2l^2}{2}(\dot q_1+\dot q_2)\sin q_2\\
\frac{m_2l^2}{2}\dot q_1\sin q_2& 0
\end{bmatrix},\nonumber
\end{equation}
\begin{equation}
G(q)=
\begin{bmatrix}
\frac{m_1gl}{2}\cos q_1\!+\!\frac{m_2gl}{2}\cos(q_1\!+\!q_2)\!+\!m_2gl\cos q_1\\
\frac{m_2gl}{2}\cos(q_1\!+\!q_2)
\end{bmatrix},\nonumber
\end{equation}
where $m_1=m_2=1\ kg$, $l=1\ m$, $q=[q_1\ q_2]\in\R^2$ denotes the joint angles, and $\omega=[ \dot q_1 \ \dot q_2]\in\R^2$ are joint angular velocities \cite{sun2011neural}. 
We emphasize that $M$, $C$, and $G$ are assumed to be unknown, and only $\la_2=5$ in Property \ref{propertymass} and $D_1=0.2$ in Assumption \ref{assumptionnoise} are available in our control design. 

\subsection{Joint Space Safe Control}
\label{sec:simj}
In this subsection, simulation results of the joint space safe control are presented. The reference trajectories are $q_{1d}=q_{2d}=3\sin (t)$; four CBFs are selected as $h_1=2.5-q_1$, $h_2=q_1+2.5$, $h_3=2-q_2$, and $h_4=q_2+1$, which aim to ensure $-2.5\leq q_1\leq 2.5$ and $-1\leq q_2\leq 2$; the control parameters are selected as $\beta=2$, $\ga=10$, $\la=16$, $\ep=\ep_1=\ep_2=0.01$, $L=0.3$, $\ga_\ta=1$, and $k_1=0.1$; the initial conditions are $q_1(0)=q_2(0)=1$ and $\dot q_1(0)=\dot q_2(0)=0$; the measurement uncertainty and disturbance are selected as $\xi=[0.2\sin(2t)\ 0.2\sin(2t)]^\top$ and $\tau_d=10\sin(t)$, from which one can see that Assumption \ref{assumptiontaud} and \ref{assumptionnoise} are satisfied. It is easy to check  that Conditions (i) and (ii) of Theorem \ref{theorem:cbf} are fulfilled with the given parameters and CBFs. The simulation results are presented in Fig. \ref{fig:joint}.

From the simulation results one can see that the safety of $\C_q$ is guaranteed as the trajectories of $q_1$ and $q_2$ always stay inside the safe region whose boundaries are represented by the dashed red line, and the reference trajectory is well-tracked within the safe set. Moreover, from Fig. \ref{fig:joint}(c) one can observe that $\|e_q(t)\|<L$ is satisfied for any $t\geq 0$, which indicates that the adaptive BLF-based controller proposed in Theorem \ref{theorem:blf} is effective.

\begin{figure}
\centering
\includegraphics[width=0.825\linewidth]{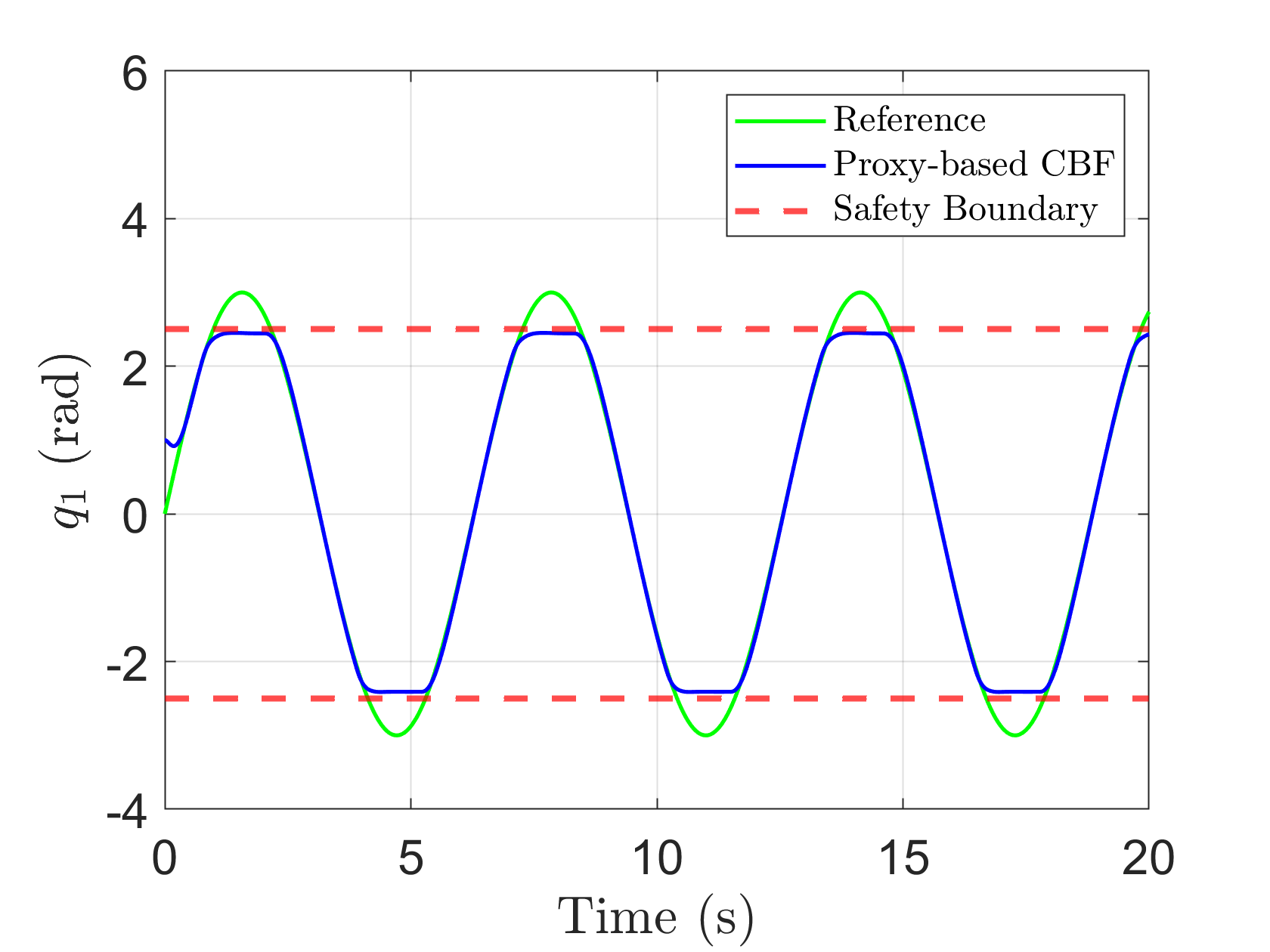}
\includegraphics[width=0.825\linewidth]{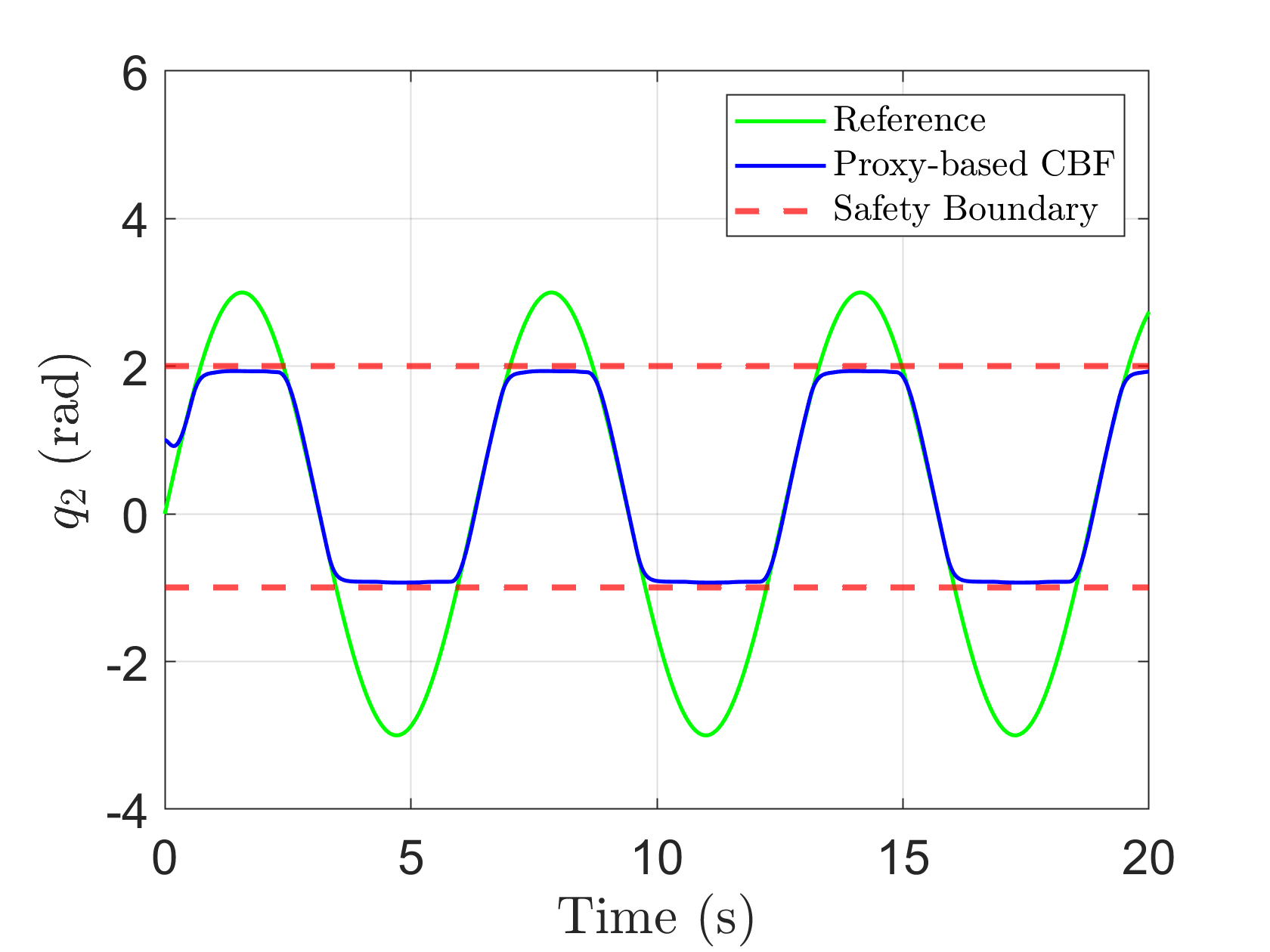}
\includegraphics[width=0.825\linewidth]{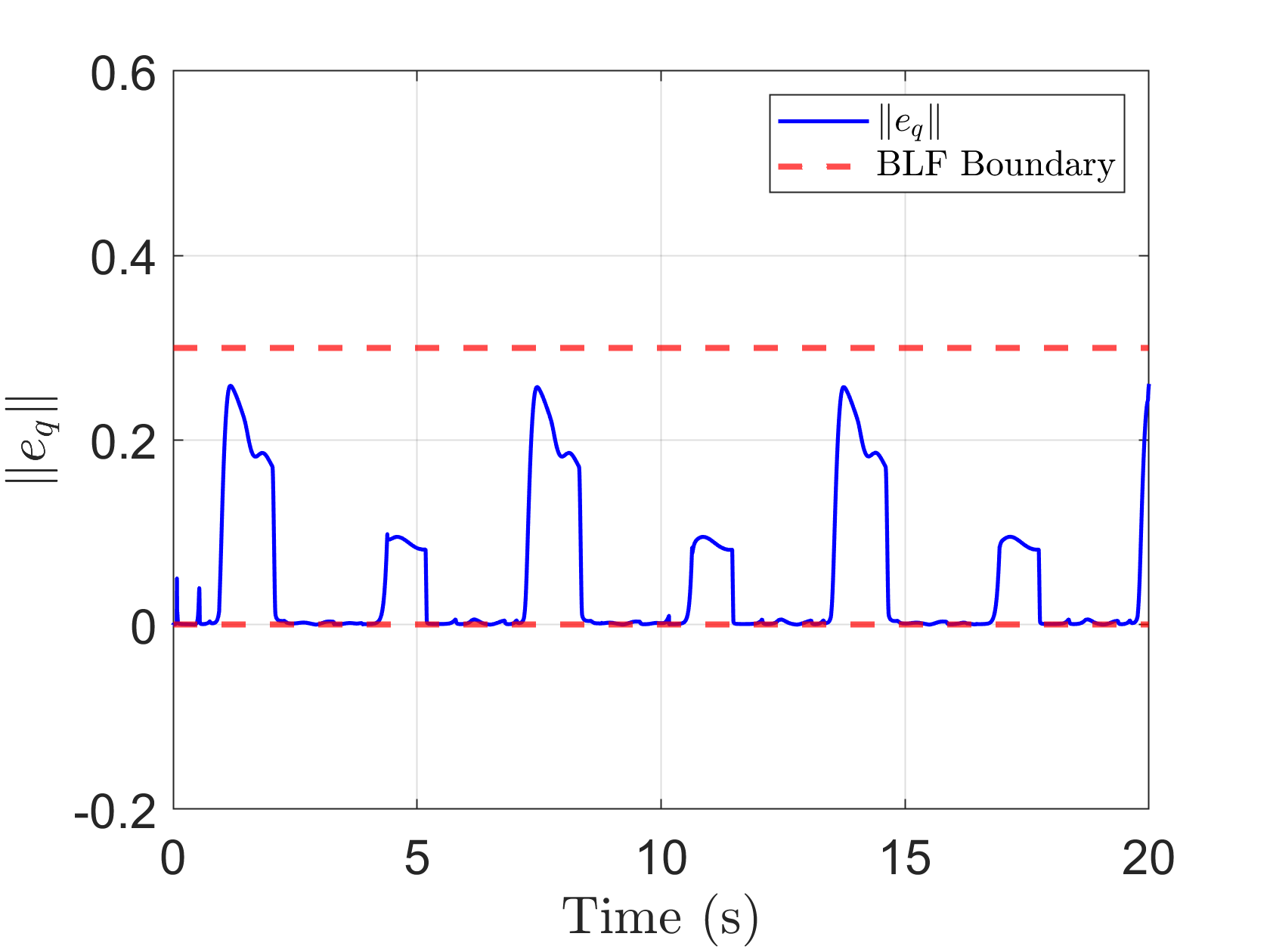}
\includegraphics[width=0.825\linewidth]{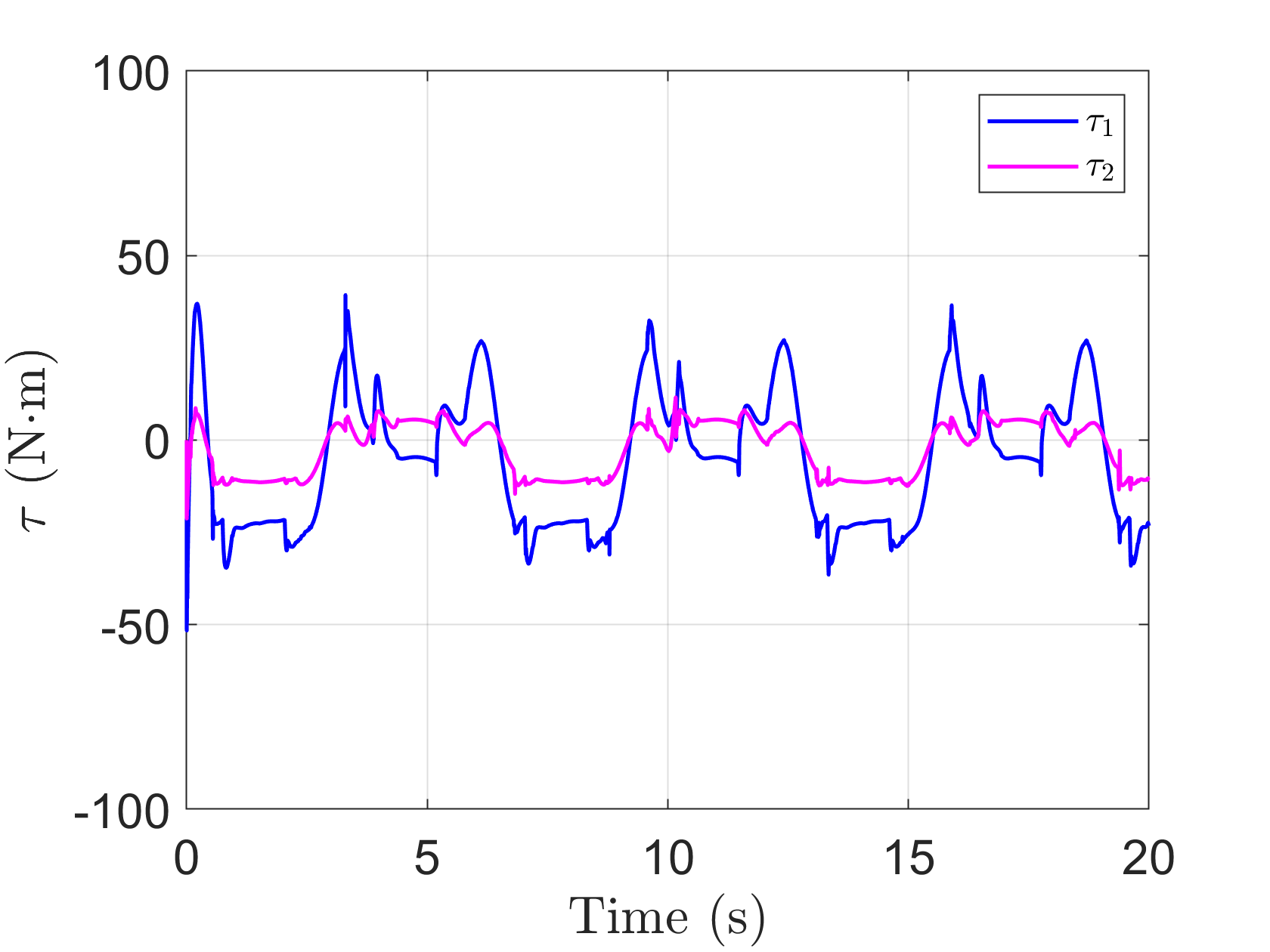}
\caption{Simulation results of the joint space safe control. From (a) and (b) it can be seen that the proposed controller can ensure safety of $\C_q$ as the trajectories of $q_1$ and $q_2$ never cross the boundary of the safe region represented by the dash red lines, with good tracking performance inside the safe region. Moreover, from (c) one can conclude that the adaptive BLF-based controller developed in Theorem \ref{theorem:blf} is effective since the constraint on $e_q$ is not violated.  }\label{fig:joint}
\end{figure}

\subsection{Task Space Safe Control}
\label{sec:simt}
In this subsection, simulation results for task space safe control are presented. The forward kinematics can be expressed as
\begin{equation*}
    \begin{bmatrix}
    x\\y
    \end{bmatrix}=\begin{bmatrix}
    l_1\cos(q_1)+l_2\cos(q_1+q_2)\\
    l_1\sin(q_1)+l_2\sin(q_1+q_2)
    \end{bmatrix},
\end{equation*}
and the Jacobian is
\begin{equation*}
    J(q)=\begin{bmatrix}
    -l_1\sin(q_1)-l_2\sin(q_1+q_2)&-l_2\sin(q_1+q_2)\\
    l_1\cos(q_1)+l_2\cos(q_1+q_2)&l_2\cos(q_1+q_2)
    \end{bmatrix}.
\end{equation*}
Note that the measurement uncertainties and disturbance are the same as those in Section \ref{sec:simj} such that Assumption \ref{assumptiontaud} and \ref{assumptionnoise} are satisfied.
To demonstrate the effectiveness of the proposed method, three cases are considered.
\begin{itemize}
    \item Case 1: The CBF is $h=x^2+y^2-0.25$; the initial conditions are  $x(0)=1.59$, $y(0)=0.11$; the reference trajectories are $x_d(t)=1.5-0.3t$, $y(t)=0$; and the control parameters are chosen as $\ep=\ep_1=\ep_2=0.01$, $L=0.05$, $\ga_\ta=1$, $\beta=2$, $\la=100$, $l_1=l_2=20$, $\ga=1000$, and $k_1=3$.
\item Case 2: The CBF is $h=1+x-y^2$; the initial conditions are  $x(0)=1.8$, $y(0)=0$; the reference trajectories are $x_d(t)=1.5\cos(t)$, $y(t)=1.5\sin(t)$; and the control parameters are the same as those in Case 1 except for $\ga=300$.
  \item Case 3: The CBF is $h=1+x+y$; the initial conditions are  $x(0)=1.8$, $y(0)=0$; the reference trajectories are $x_d(t)=1.5\cos(t)$, $y(t)=1.5\sin(t)$; and the control parameters are the same as those in Case 1 except for $l_1=l_2=40$, $\la=100$, and $\ga=500$.
\end{itemize}
The simulation results are presented in Fig. \ref{fig:task}, from which one can see that in all three cases the safety of $\C_p$ is ensured by the proposed controller as the trajectories of $x$ and $y$ always stay inside the safe region whose boundary is represented by the dash red lines, and the tracking performance inside the safe region is satisfactory.

\begin{figure}[!t]
\centering
\includegraphics[width=\linewidth]{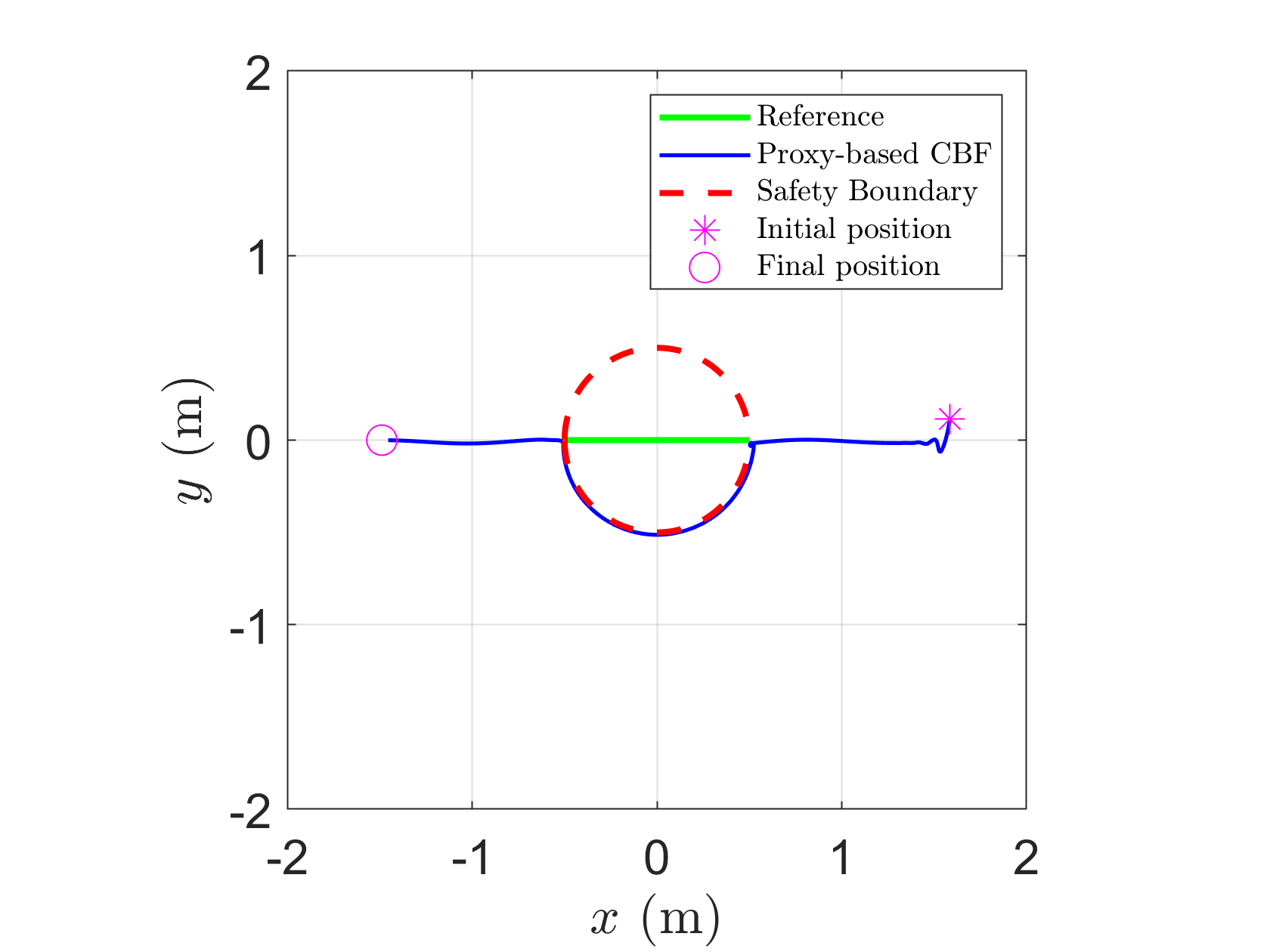}\vskip 2mm
\includegraphics[width=\linewidth]{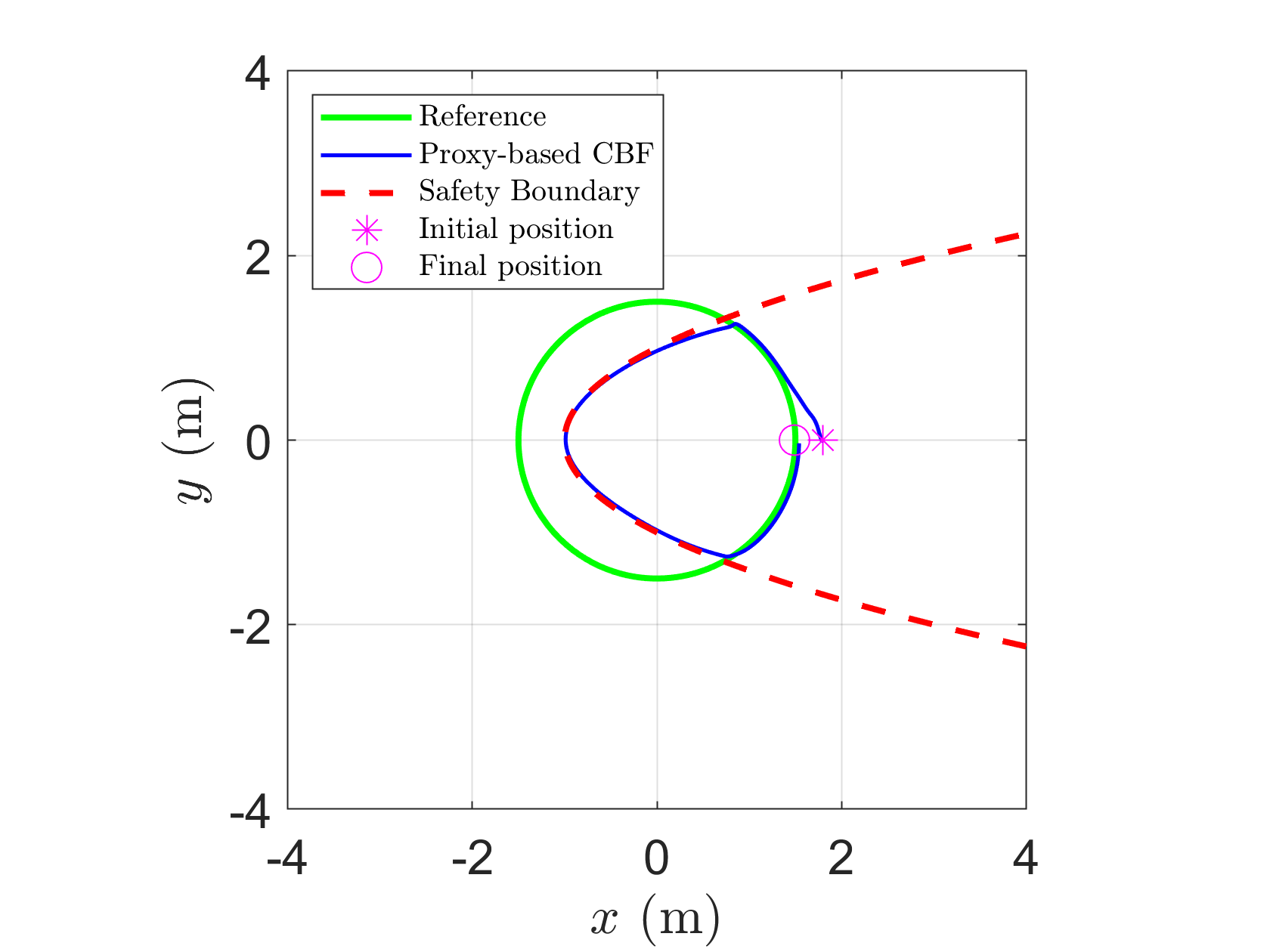}\vskip 2mm
\includegraphics[width=\linewidth]{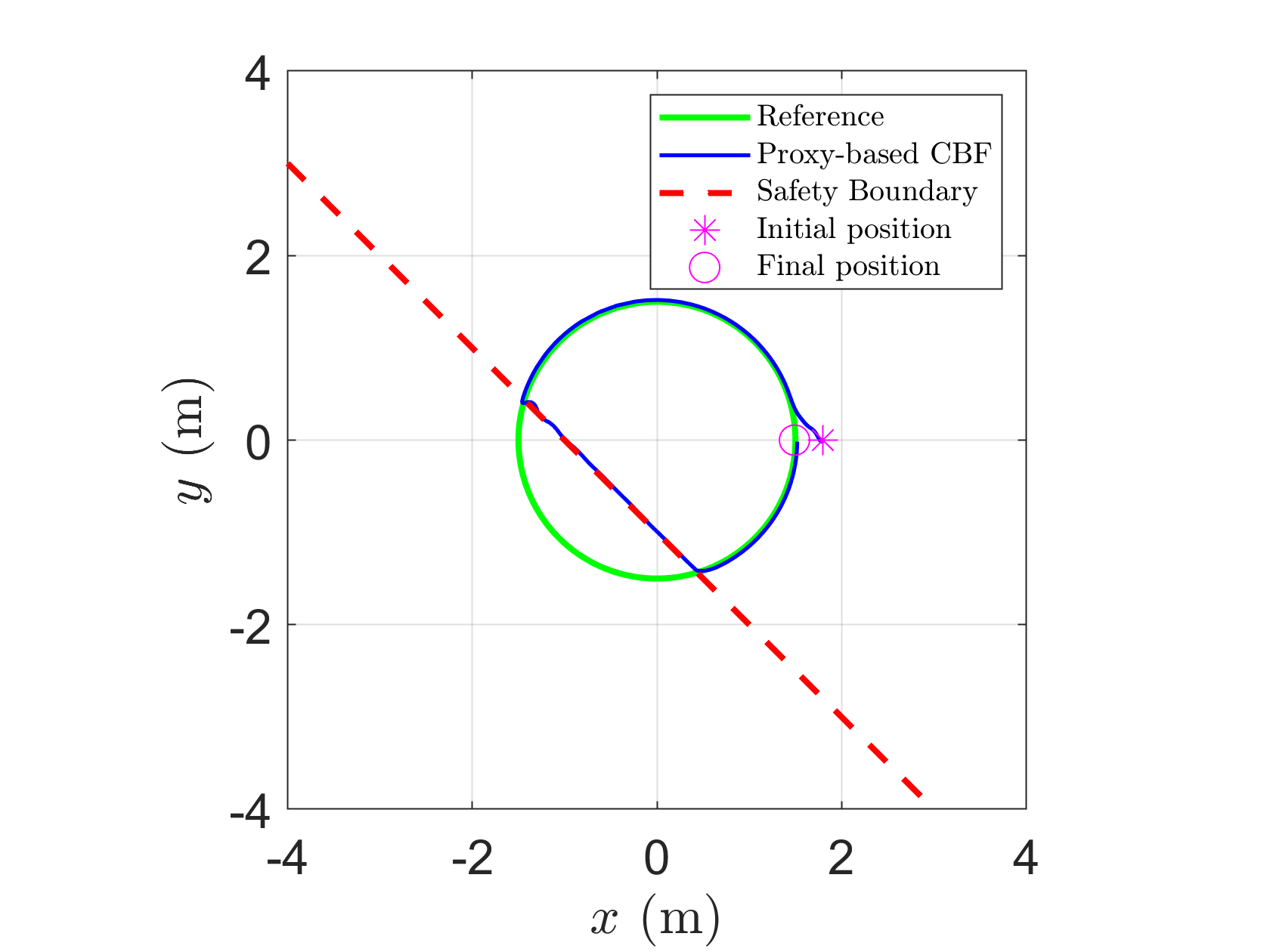}\vskip 2mm
\caption{Simulation results of the task space safe control. One can observe that proposed controller can ensure safety of $\C_p$ as the trajectories of $x$ and $y$ always stay inside the safe region whose boundary is represented by the dash red lines, and the tracking performance inside the safe region is satisfactory. }\label{fig:task}
\end{figure}

\section{Conclusion}
\label{sec:conclusion}
In this paper, a novel proxy CBF-BLF-based control design approach is proposed for EL systems with limited information by decomposing an EL system into the proxy subsystem and the virtual tracking subsystem. A BLF-based controller is designed for the virtual tracking subsystem to ensure the boundedness of the safe velocity tracking error. Based on that, a CBF-based controller is designed for the proxy subsystem to ensure safety in the joint space or task space. Simulation results are given to verify the effectiveness of the proposed method. Future work includes conducting experimental studies and generalizing the results to ensure safety and stability simultaneously for EL systems.

%=====================================================================
%            Appendix                                                 |
%=====================================================================

\bibliographystyle{IEEEtran}
\bibliography{CDC2023}
\end{document}